\documentclass[envcountsame]{llncs}
  
\usepackage{svn-multi}
\usepackage{amsmath}
\usepackage{amssymb}
\usepackage{verbatim} 
\usepackage{stmaryrd}
\usepackage{paralist}
\usepackage{cite}
\usepackage{appendix}
\usepackage{url}
\usepackage[]{units}
\usepackage{tikz}
\usepackage{subfig}
\usepackage{array}
\usepackage{color}
\usepackage{empheq}
\usetikzlibrary{calc,arrows,positioning}
\allowdisplaybreaks

\newif\ifdraft\drafttrue

\DeclareFontFamily{U}{mathb}{\hyphenchar\font45}
\DeclareFontShape{U}{mathb}{m}{n}{
      <5> <6> <7> <8> <9> <10> gen * mathb
      <10.95> mathb10 <12> <14.4> <17.28> <20.74> <24.88> mathb12
}{}
\DeclareSymbolFont{mathb}{U}{mathb}{m}{n}
\DeclareMathSymbol{\sqdoublecup} {2}{mathb}{"5F} 
\DeclareMathSymbol{\boxplus} {2}{mathb}{"60} 

\usepackage[amsmath]{ntheorem}
\makeatletter
\renewtheoremstyle{plain}%
  {\item[\hskip\labelsep \theorem@headerfont ##1\ ##2\theorem@separator]}%
  {\item[\hskip\labelsep \theorem@headerfont ##1\ ##2, ##3\theorem@separator]}
\makeatother

\newcommand{\Pow}[1]{\mathop{\mbox{$\mathcal P$}} ({#1})   } 
\newcommand{\dist}[1]{\mathop{\mbox{$\mathcal D$}} ({#1})   } 

\newcommand{\support}[1]{\lceil{#1}\rceil}
\newcommand{\eDis}{\mathop{\varepsilon}}
\newcommand{\Op}[1]{[\![#1]\!]}

\newcommand{\trans}[1][]{\xrightarrow{\, {#1} \, }}

\newcommand{\Act}{A}

\newcommand{\const}[1]{\underline{#1}}
\newcommand{\dia}[1]{\langle #1\rangle}
\newcommand{\pdist}[1]{\overline{#1}}
\newcommand{\tdists}{\it tdists}
\newcommand{\sem}[2]{\llbracket{#1}\rrbracket({#2})}
\newcommand{\rplus}{\oplus}

\DeclareMathOperator{\Kantorovich}{\mathit{K}}
\DeclareMathOperator{\Hausdorff}{\mathit{H}}

\DeclareMathOperator{\der}{\mathit{der}}
\DeclareMathOperator{\cc}{\mathit{cc}}

\newcommand{\tr}{{\it tr}}
\newcommand{\Prob}{{\it Pr}}
\newcommand{\strong}{\ensuremath{\mathit{sb}}}
\newcommand{\convex}{\ensuremath{\mathit{cb}}}

\DeclareMathOperator{\traced}{\mathbf{d}}
\DeclareMathOperator{\bisimd}{\mathbf{d}}
\DeclareMathOperator{\tracedist}{\traced_{\it t}}
\DeclareMathOperator{\tracedistl}{\traced_{\it t}^l}
\DeclareMathOperator{\bisimddist}{\bisimd_{\it db}}
\DeclareMathOperator{\bisimddistl}{\bisimd_{\it db}^{ld}}
\DeclareMathOperator{\bisimdstrong}{\bisimd_{\strong}}
\DeclareMathOperator{\bisimdconvex}{\bisimd_{\convex}}

\DeclareMathOperator{\bisimdlogicstrong}{\bisimd_{\strong}^{ls}}
\DeclareMathOperator{\bisimddlogicstrong}{\bisimd_{\strong}^{ld}}
\DeclareMathOperator{\bisimdlogicconvex}{\bisimd_{\convex}^{l}}
\DeclareMathOperator{\bisimdtop}{\mathbf{1}}
\DeclareMathOperator{\bisimdbot}{\mathbf{0}}

\DeclareMathOperator{\functor}{\mathit{F}}

\DeclareMathOperator{\functorconvex}{\functor_\convex}
\DeclareMathOperator{\functorany}{\functor}

\DeclareMathOperator{\logic}{\mathcal{L}}
\DeclareMathOperator{\logicstate}{\logic^{S}}
\DeclareMathOperator{\logicdist}{\logic^{D}}

\newcommand{\leaveout}[1]{}

\usepackage[ddmmyyyy]{datetime}
\usepackage{datenumber}
\usepackage{expl3}
\usepackage{color}

\usepackage{etex}
\reserveinserts{30}

\let\doendproof\endproof
\renewcommand\endproof{~\hfill\qed\doendproof} 
\let\doendexample\endexample
\newenvironment{apx-proof}[1] 
        {\noindent \rm \textbf{Proof of #1.}} 
        {\qed}
\renewcommand\endexample{~\hfill$\qed$\doendexample}
\renewcommand{\figurename}{Figure}

\makeatletter
\newcommand*\getnumtz[2]{%
    \expandafter\@getnumtz\the\numexpr 0#2\relax
        \empty\relax\relax\@nnil{#1}{#2}%
}

\def\@getnumtz#1\relax#2\relax#3\@nnil#4#5{%
    \ifx\relax#2\relax
        \edef#4{#1}%
    \else
        \begingroup\expandafter\endgroup
        \expandafter\let\expandafter#4\csname getnumtz@#5\endcsname%
    \fi
}

\newcommand*\definetz[2]{%
    \@namedef{getnumtz@#1}{#2}%
}%

\definetz{Z}{+0000}
\definetz{GMT}{+0000}
\definetz{UTC}{+0000}
\definetz{CET}{+0100}
\definetz{CEST}{+0200}

\newcommand*\converttimezone[9]{%
    \begingroup
    \c@myyear=\numexpr#2\relax
    \c@mymonth=\numexpr#3\relax
    \c@myday=\numexpr#4\relax
    \c@myhour=\numexpr#5\relax
    \c@myminute=\numexpr#6\relax
    \c@mysecond=\numexpr#7\relax
    \getnumtz\origtz{#8}%
    \getnumtz\targettz{#9}%
    \c@myhourminute=\numexpr (#5)*100+(#6) - \origtz + \targettz \relax
    \c@myhour=\numexpr \c@myhourminute / 100\relax
    \c@myminute=\numexpr \c@myhourminute - \c@myhour*100\relax
    \loop\ifnum\c@myminute<\z@
        \advance\c@myhour by \m@ne
        \advance\c@myminute by 60\relax
    \repeat
    \loop\ifnum\c@myminute>59\relax
        \advance\c@myhour by \@ne
        \advance\c@myminute by -60\relax
    \repeat
    \ifnum\c@myhour<0\relax
        \setmydatenumber{mydatenumber}{\value{myyear}}{\value{mymonth}}{\value{myday}}%
        \advance\c@mydatenumber by \m@ne
        \setmydatebynumber{\value{mydatenumber}}{myyear}{mymonth}{myday}%
        \advance\c@myhour by 24\relax
    \else\ifnum\c@myhour>23\relax
        \setmydatenumber{mydatenumber}{\value{myyear}}{\value{mymonth}}{\value{myday}}%
        \advance\c@mydatenumber by \@ne
        \setmydatebynumber{\value{mydatenumber}}{myyear}{mymonth}{myday}%
        \advance\c@myhour by -24\relax
    \fi\fi
    \edef\@tempa{\unexpanded{#1}{\themyyear}{\themymonth}{\themyday}{\themyhour}{\themyminute}{\themysecond}{#9}}%
    \expandafter
    \endgroup\@tempa
}

\newcounter{myyear}
\newcounter{mymonth}
\newcounter{myday}
\newcounter{myhour}
\newcounter{myminute}
\newcounter{mysecond}
\newcounter{mydatenumber}
\makeatother

\ExplSyntaxOn
\cs_set_eq:NN \lgt_file_date: \today

\cs_new_nopar:Npn \lgt_file_time:
  {
    \lgt_get_time:N \currenthour :
    \lgt_get_time:N \currentminute :
    \lgt_get_time:N \currentsecond
  }

\cs_new_nopar:Npn \lgt_get_time:N #1
  {
    \int_compare:nNnT { #1 } < { 10 } { 0 }
    \int_use:N #1
  }

\ExplSyntaxOff

\definecolor{lightblue}{RGB}{224,224,255}
\definecolor{lightred}{RGB}{255,224,224}
\definecolor{lightgreen}{RGB}{224,255,224}
\definecolor{lightyellow}{RGB}{255,255,224}
\definecolor{lightpurple}{RGB}{255,224,255}
\definecolor{darkerred}{RGB}{64,0,0}
\definecolor{darkred}{RGB}{128,0,0}
\definecolor{darkblue}{RGB}{0,0,128}
\definecolor{darkgreen}{RGB}{0,128,0}
\definecolor{darkpurple}{RGB}{128,0,128}
\definecolor{grey}{rgb}{0.745098,0.745098,0.745098}
\definecolor{lightgrey}{rgb}{0.9,0.9,0.9}
\definecolor{darkgrey}{rgb}{0.6,0.6,0.6}

\newcommand{\colorpar}[3]{\colorbox{#1}{\parbox{#2}{#3}}}

\newcommand{\marginremark}[3]{\marginpar{\colorpar{#2}{\linewidth}{\color{#1}#3}}}

\makeatletter
\def\THICKhrulefill{\leavevmode \leaders \hrule height 5pt\hfill \kern \z@}
\makeatother

\newcommand{\remarkDG}[1]{\marginremark{darkred}{lightred}{\tiny{[DG]~ #1}}}
\newcommand{\remarkYD}[1]{\marginremark{darkgreen}{lightgreen}{\tiny{[YD]~ #1}}}

\ifdraft
\else
\renewcommand{\remarkDG}[1]{}
\renewcommand{\remarkYD}[1]{}
\fi

\begin{document}

\title{Modal Characterisations of Behavioral Pseudometrics}

\author{
Yuxin Deng\inst{1} \and
Wenjie Du\inst{2} \and
Daniel Gebler\inst{3} 
}

\institute{
East China Normal University \and
Shanghai Normal University \and
VU University Amsterdam 
\\ \today
}

\maketitle

\begin{abstract}
For the model of probabilistic labelled transition systems that allow for the co-existence of nondeterminism and probabilities, we present two notions of bisimulation metrics: one is state-based and the other is distribution-based. We provide a sound and complete modal characterisation for each of them, using real-valued modal logics based on the Hennessy-Milner logic. The logic for characterising the state-based metric is much simpler than an earlier logic by Desharnais et al. as it uses only two non-expansive operators rather than the general class of non-expansive operators.
\end{abstract}


\section{Introduction}\label{sec:introduction}
Bisimulation is an important proof technique for establishing behavioural equivalences of concurrent systems.
In probabilistic concurrency theory, there are roughly two kinds of bisimulations: one is state-based because it is directly defined over states and then lifted to distributions, and the other is distribution-based as it is a relation between distributions. The former is originally defined in \cite{LS91} to represent a branching time semantics; the latter, as defined in \cite{HKK14,FZ14,DFD15}, is strictly coarser and represents a linear time semantics. 

In correspondence with those bisimulations, there are two notions of behavioural pseudometrics (simply called metrics). They are more robust ways of formalising behavioural similarity between formal systems than bisimulations because, particularly in the probabilistic setting, bisimulations are too sensitive to probabilities (a very small perturbation of the probabilities would render two systems non-bisimilar). A metric gives a quantitative measure of the distance between two systems and distance $0$ usually means that the two systems are bisimilar.  A logical characterisation of the state-based bisimulation metric for labelled Markov processes is given in \cite{DGJP04}. For a more general model of labelled concurrent Markov chains (LCMCs) that allow for the co-existence of nondeterminism and probabilities, a weak bisimulation metric is proposed in \cite{DJGP02}. Its logical characterisation uses formulae like $h\circ f$, where $f$ is a formula and $h$ can be any non-expansive operator on $[0,1]$, i.e. $|h(x)-h(y)|\leq |x-y|$ for any $x,y\in [0,1]$. 
A natural question then arises: instead of the general class of non-expansive operators, is it possible to use only a few simple non-expansive operators without losing the capability of characterising the bisimulation metric?

In the current work, we give a positive answer to the above question.
We work in the framework of probabilistic labelled transition systems (pLTSs) that are essentially the same as LCMCs, so the interplay of nondeterminism and probabilities is allowed. We provide a modal characterisation of the state-based bisimulation metric closely in line with the classical Hennessy-Milner logic (HML) \cite{HM85}. Our variant of the HML makes use of state formulae and distribution formulae, which are formulae evaluated at states and distributions, respectively, and yield success probabilities. We use merely two non-expansive operators: negation ($\neg\phi$) and testing ($\phi\ominus p$). Negation is self-explanatory and the testing operator checks if a state satisfies a property with certain threshold probability. More precisely, if state $s$ satisfies formula $\phi$ with probability $q$, then it satisfies $\neg\phi$ with probability $1-q$, and  satisfies $\phi\ominus p$ with probability $q-p$ if $q>p$ and $0$ otherwise. In other words, we do not need the general classs of non-expansive operators because negation and testing, together with other modalities in the classical HML, are expressive enough to characterise bisimulation metrics. As regards to the characterisation of distribution-based bisimulation metric, we drop state formulae and  use distribution formulae only.

The rest of this paper is organised as follows. Section~\ref{sec:preliminaries} provides some basic concepts on pLTSs. Section~\ref{sec:bisim_metric} defines a two-sorted modal logic that leads to a sound and complete characterisation of the state-based bisimulation metric. Section~\ref{sec:dist} gives a similar characterisation for the distribution-based bisimulation metric. In Section~\ref{sec:relwork} we review some related work.
 Finally, we conclude in Section~\ref{sec:conclu}.

\section{Preliminaries}\label{sec:preliminaries}
Let $S$ be a countable set. A \emph{(discrete) probability subdistribution} over $S$ is a function $\Delta:S\rightarrow [0,1]$ with $\sum_{s\in S}\Delta(s)\leq 1$. It is a \emph{(full) distribution} if $\sum_{s\in S}\Delta(s)= 1$.
Let $\dist{S}$ denote the set of all distributions over $S$. A matching $\omega \in \dist{S \times S}$ for $(\Delta,\Theta) \in \dist{S} \times \dist{S}$ is given if $\sum_{t\in S} \omega(s,t)=\Delta(s)$ and $\sum_{s\in S} \omega(s,t)=\Theta(t)$ for all $s,t\in S$. We denote the set of all matchings for $(\Delta,\Theta)$ by $\Omega(\Delta,\Theta)$. 

A \emph{metric} $d$ over space $S$ is a distance function $d: S\times S\rightarrow\mathbb{R}_{\geq 0}$ satisfying: (i) $d(s,t)=0$ iff $s=t$ (isolation), (ii) $d(s,t)=d(t,s)$ (symmetry), (iii) $d(s,t)\leq d(s,u) + d(u,t)$ (triangle inequality), for any $s,t,u\in S$. If we replace (i) with $d(s,s)=0$ for all $s\in S$, we obtain a \emph{pseudometric}. In this paper we are interested in pseudometrics because two distinct states can still be at distance zero if their behaviour is similar. But for simplicity, we often use metrics for pseudometrics. A metric $d$ over $S$ is $c$-bounded if $d(s,t)\leq c$ for any $s,t\in S$, where $c\in \mathbb{R}_{\geq 0}$ is a positive real number.

Let $d\colon S \times S \to [0,1]$ be a metric over $S$. We lift it to be a metric over $\dist{S}$ by using the \emph{Kantorowich metric} \cite{KR58} $\Kantorovich(d)\colon \dist{S} \times \dist{S} \to [0,1]$ defined via a linear programming problem as follows:
\[
	\Kantorovich(d)(\Delta,\Theta) = \min_{\omega \in \Omega(\Delta,\Theta)} \sum_{s,t\in S}d(s,t) \cdot \omega(s,t)
\]
for $\Delta,\Theta \in \dist{S}$. 
The dual of the above linear programming problem is the following
\[\begin{array}{rl}
\max\ \sum_{s\in S}(\Delta(s)-\Theta(s)) x_s, \text{ subject to } & 0\leq x_s\leq 1 \\ 
&\forall s,t\in S\colon\ x_s-x_{t}\leq d(s,t)\;.
\end{array}\]
The duality theorem in linear programming guarantees that both problems have the same optimal value.

Let $\hat{d}\colon \dist{S} \times \dist{S} \to [0,1]$ be a metric over $\dist{S}$. We lift it to be a metric over the powerset of $\dist{S}$, written $\Pow{\dist{S}}$, by using the \emph{Hausdorff metric} $\Hausdorff(\hat{d})\colon \Pow{\dist{S}} \times \Pow{\dist{S}} \to [0,1]$ given as follows
\[
	\Hausdorff(\hat{d})(\Pi_1,\Pi_2) = 
		\max \{ 
			\adjustlimits\sup_{\Delta \in \Pi_1}\inf_{\Theta \in \Pi_2} \hat{d}(\Delta,\Theta),\ \
			\adjustlimits\sup_{\Theta\in \Pi_2}\inf_{\Delta\in \Pi_1} \hat{d}(\Theta,\Delta) 
		\}
\]
for all $\Pi_1,\Pi_2 \subseteq \dist{S}$, whereby $\inf \emptyset = 1$ and $\sup \emptyset = 0$.

Probabilistic labelled transition systems (pLTSs) generalize labelled transition systems (LTSs) by allowing for probabilistic choices in the transitions. We consider pLTSs (or essentially \emph{simple probabilistic automata} \cite{Seg95a}) with countable state spaces.
\begin{definition}
A \emph{probabilistic labelled transition system} is a triple $(S,\Act,{\trans})$, where $S$ is a countable set of states, $\Act$ is a countable set of actions, and the relation ${\trans} \subseteq S \times \Act \times \dist{S}$ is a transition relation.
\end{definition}
We write $s \trans[a] \Delta$ for ${(s,a,\Delta)} \in {\trans}$ and define $\der(s,a) = \{\Delta \mid s \trans[a] \Delta \}$ as the set of all $a$-successor distributions of $s$. A pLTS is \emph{image-finite} if for any state $s$ and action $a$ the set $\der(s,a)$ is finite. In the current work, we focus on image-finite pLTSs.

\section{State-Based Bisimulation Metrics}\label{sec:bisim_metric}


We consider the complete lattice $([0, 1]^{S\times S},\sqsubseteq)$ defined by  $d \sqsubseteq d'$ iff $d(s, t) \le d'(s, t)$, for all $s, t \in S$. 
For some $D \subseteq [0, 1]^{S\times S}$ the least upper bound is given by $(\bigsqcup D)(s,t) = \sup_{d\in D}d(s,t)$, and the greatest lower bound is given by $(\bigsqcap D)(s,t) = \inf_{d\in D}d(s,t)$ for all $s, t \in S$. The bottom element $\bisimdbot$ is the constant zero function $\bisimdbot(s,t)=0$ and the top element $\bisimdtop$ is the constant one function $\bisimdtop(s,t)=1$ for all $s,t \in S$. 


\begin{definition}
\label{def:sb}
A $1$-bounded metric $d$ on $S$ is a \emph{state-based bisimulation metric} if for all $s,t\in S$ and $\epsilon \in [0,1)$ with $d(s,t) \le \epsilon$, if $s \trans[a] \Delta$ then there exists some $t \trans[a] \Delta'$ with $\Kantorovich(d)(\Delta,\Delta') \le \epsilon$.
\end{definition}
The smallest (wrt. $\sqsubseteq$) state-based bisimulation metric, denoted by $\bisimdstrong$, is called \emph{ state-based bisimilarity metric}. 
Its kernel is the state-based bisimilarity as defined in \cite{LS91,Seg95a}.

\begin{figure}[!t]
  \begin{center}
  \scalebox{0.5}{
\begin{tikzpicture}[->,>=stealth,auto,node distance=1.7cm,semithick,scale=1,every node/.style={scale=1}]
	\tikzstyle{state}=[minimum size=25pt,circle,draw,thick]
        \tikzstyle{triangleState}=[minimum size=0pt,regular polygon, regular polygon sides=4,draw,thick]
	\tikzstyle{stateNframe}=[]
	every label/.style=draw
        \tikzstyle{blackdot}=[circle,fill=black, minimum
        size=6pt,inner sep=0pt]
      \node[state](s){$s$};
      \node[state](s1)[below of=s]{$s_1$};
      \node[blackdot](m1)[below of=s1]{};
      \node[state](s2)[below  of=m1,xshift=-1cm]{$s_2$};
      \node[state](s3)[below  of=m1,xshift=1cm]{$s_3$};
      \node[state](s4)[below of=s2,xshift=1cm]{$s_4$};

      \node[state](t)[left of=s,xshift=10cm]{$t$};
      \node[blackdot](n1)[below of=t]{};
      \node[state](t1)[below of=n1,xshift=-1cm]{$t_1$};
      \node[state](t2)[below of=n1,xshift=1cm]{$t_2$};
      \node[state](t3)[below of=t1]{$t_3$};
      \node[state](t4)[below of=t2]{$t_4$};
      \node[state](t5)[below of=t3,xshift=1cm]{$t_5$};

     \path (s) edge     node[right] {$a$} (s1)
  	   (s1) edge[-]     node[right] {$b$} (m1)
	   (m1) edge[dashed]     node[left] {$\frac{1}{2}$}   (s2)
                edge[dashed]     node[right] {$\frac{1}{2}$}   (s3)
           (s2) edge             node[left] {$c$} (s4)
           (s3) edge             node[right] {$d$} (s4)
           (t) edge[-]             node {$a$} (n1)
           (n1) edge[dashed]             node[left] {$\frac{1}{2}$} (t1)
                edge[dashed]             node[right] {$\frac{1}{2}$} (t2)
  	   (t1) edge     node[left] {$b$} (t3)
  	   (t2) edge     node[right] {$b$} (t4)
  	   (t3) edge     node[left] {$c$} (t5)
  	   (t4) edge     node[right] {$d$} (t5);
\end{tikzpicture}
}
\end{center}
\caption{\label{fig:exam1} $\bisimdstrong(s,t)=\frac{1}{2}$}
\end{figure}
\begin{example}\label{ex:exam1}
In this example, we calculate the distance between states $s$ and $t$ in Figure~\ref{fig:exam1}. Firstly, observe that $\bisimdstrong(s_2,t_3)=0$ because $s_2$ is bisimilar to $t_3$ while $\bisimdstrong(s_3,t_3)=1$ because the two states $s_3$ and $t_3$ perform completely different actions. Secondly, let $\Delta=\frac{1}{2}\pdist{s_2}+\frac{1}{2}\pdist{s_3}$ and $\Theta=\pdist{t_3}$. We see that 
\[\begin{array}{rcl}
\Kantorovich(\bisimdstrong)(\Delta,\Theta) & = & \min_{\omega \in \Omega(\Delta,\Theta)} \bisimdstrong(s_2,t_3) \cdot \omega(s_2,t_3) + \bisimdstrong(s_3,t_3) \cdot \omega(s_3,t_3)\\
& = & \min_{\omega \in \Omega(\Delta,\Theta)} 1 \cdot \omega(s_2,t_3) + 0 \cdot \omega(s_3,t_3)\\
& = & \min_{\omega \in \Omega(\Delta,\Theta)} 1 \cdot \frac{1}{2} + 0 \cdot \frac{1}{2}\\
& = & \frac{1}{2}
\end{array}\]
It follows that $\bisimdstrong(s_1,t_1)=\frac{1}{2}$. Similarly, we get $\bisimdstrong(s_1,t_2)=\frac{1}{2}$. Then it not difficult to see that
\[\Kantorovich(\bisimdstrong)(\pdist{s_1},\frac{1}{2}\pdist{t_1}+\frac{1}{2}\pdist{t_2}) ~=~ \bisimdstrong(s_1,t_1)\cdot\frac{1}{2} + \bisimdstrong(s_1,t_2)\cdot\frac{1}{2} ~=~ \frac{1}{2}\]
from which we finally obtain $\bisimdstrong(s,t)=\frac{1}{2}$.
\end{example}

The above coinductively defined bisimilarity metric can be reformulated as a fixed point of a monotone functor. Let us define the functor $\functorany\colon [0,1]^{S \times S} \to [0,1]^{S \times S}$ for $d\colon S \times S \to [0,1]$ and $s,t\in S$ by
\[
	\functorany(d)(s,t) = \sup_{a\in \Act} \{ \Hausdorff(\Kantorovich(d))(\der(s,a), \der(t,a)) \}\,.
\]
%
It can be shown that $\functorany$ is monotone and its least fixed point is defined by $\bigsqcap d_i$, where $d_0 = \bisimdbot$ and $d_{i+1} = \functorany(d_i)$ for all $i\in\mathbb{N}$.

\begin{proposition}
For image-finite pLTSs, the strong bisimilarity metric is the least fixed point of $\functorany$.
\hfill\qed
\end{proposition}


Now we proceed by defining a real-valued modal logic based on the Hennessy-Milner logic \cite{HM85}, called metric HML, to characterize the bisimilarity metric. Our logic is motivated by \cite[Def.~4.1]{DJGP02}, \cite[Def.~4.1]{DGJP04} and \cite[Sec.~4]{DTW10}. In the remainder of this section we confine ourselves to pLTSs with finitely many states.

\begin{definition}\label{def:mHML}
Our metric  HML is two-sorted and has the following syntax:
\[
\begin{array}{rl}
	\varphi ::= & \top \ |\ \lnot\varphi \ |\ \varphi\ominus p \ |\ \varphi_1 \land \varphi_2 \ | \ \dia{a}\psi \\[0.5ex]

	\psi ::= & \psi\ominus p \ | \ \psi_1 \land \psi_2\ | \  [\varphi] 
\end{array}
\]

with $a \in A$ and $p\in [0,1]$. 
\end{definition}
Let $\logic$ denote the set of all metric HML formulae, $\varphi$ range over the set of all \emph{state formulae} $\logicstate$, and $\psi$ range over the set of all \emph{distribution formulae} $\logicdist$. The two kinds of formulae are defined simultaneously. The operator $\phi\ominus p$ tests if a state passes $\phi$ with probability at least $p$. If $\phi$ is a state formula then it immediately induces a distribution formula $[\phi]$.
Sometimes we abbreviate $\dia{a}[\varphi]$ as $\dia{a}\varphi$. All other operators are standard and have appeared in the classical HML.



\begin{definition}
A state formula $\varphi \in \logicstate$ evaluates in $s \in S$ as follows:
\begin{itemize}[\label={}]
	\item $\displaystyle \sem{\top}{s}=1$,\\[-0.5ex]
	\item $\displaystyle \sem{\lnot \varphi}{s}=1-\sem{\varphi}{s}$,\\[-0.5ex]
	\item $\displaystyle \sem{\varphi\ominus p}{s}=\max(\sem{\varphi}{s}-p,~ 0)$,\\[-0.5ex]
	\item $\displaystyle \sem{\varphi_1 \land \varphi_2}{s}=\min(\sem{\varphi_1}{s},\sem{\varphi_2}{s})$,\\[-0.5ex]
	\item $\displaystyle \sem{\dia{a}\psi}{s}=\max_{s \trans[a] \Delta} \sem{\psi}{\Delta}$,\\[-0.5ex]
\end{itemize}
and a distribution formula $\psi \in \logicdist$ evaluates in $\Delta \in \dist{S}$ as follows: 
\begin{itemize}[\label={}]
       \item $\displaystyle \sem{\psi \ominus p}{\Delta}=\max(\sem{\psi}{\Delta}-p,0)$,\\[-0.5ex]

	\item $\displaystyle \sem{\psi_1 \land \psi_2}{\Delta}=\min(\sem{\psi_1}{\Delta},\sem{\psi_2}{\Delta})$,\\[-0.5ex]
	\item $\displaystyle \sem{[\varphi]}{\Delta}=\sum_{s\in S}\Delta(s)\sem{\varphi}{s}$. \\[-0.5ex]
%
\end{itemize}
\end{definition}
%
We often use constant formulae e.g. $\const{p}$ for any $p\in [0,1]$ with the semantics $\sem{\const{p}}{s} = p$, which is derivable in the above logic by letting $\const{p}=\top \ominus (1-p)$. Moreover, we write $\varphi\rplus p$ for $\lnot ((\lnot \varphi) \ominus p)$ which has the semantics $\displaystyle \sem{\varphi\rplus p}{s}=\min(\sem{\varphi}{s}+p,~1)=1-\max(1-\sem{\varphi}{s}-p,0)$. In the presence of negation and conjunction we can derive disjunction by letting $\varphi_1\vee\varphi_2$ to be $\lnot(\lnot \varphi_1\wedge\lnot\varphi_2)$.
Semantics of $\dia{a}\varphi$ is a translation of \cite[Def.~4.1]{DJGP02} from labelled concurrent Markov chain semantics to pLTSs. Conjunction of distribution formulae $\psi_1 \land \psi_2$ could alternatively be replaced by $\psi_1 \oplus_p \psi_2$ or $\psi_1 \oplus \psi_2$ \cite[Sec.~4]{Hen12} with semantics
\[
\begin{array}{rl}
	\sem{\psi_1 \oplus_p \psi_2}{\Delta} &= \displaystyle \sup_{\Delta = p\Delta_1 + (1-p)\Delta_2} (p\cdot\sem{\psi_1}{\Delta_1} + (1-p)\sem{\psi_2}{\Delta_2} \\
	\sem{\psi_1 \oplus \psi_2}{\Delta} &= \displaystyle \sup_{p \in (0,1)} \sem{\psi_1 \oplus_p \psi_2}{\Delta}
\end{array}
\]
\leaveout{
This uses the ``splitting property'' of metric distances between distributions \cite[Lemma~3.5 and Corollary~3.6]{DJGP02}. 

\begin{proposition}
Let $r\colon \logicdist \to \logicdist$ be defined by structural induction:
\[
	r(\psi) =
	\begin{cases}
		r(\psi_1 \ominus p) \land r(\psi_2 \ominus p) & \text{if }\psi=(\psi_1 \land \psi_2) \ominus p \\
		[\varphi \ominus p/|S|] & \text{if }\psi=[\varphi] \ominus p \\
		r(\psi_1) \land r(\psi_2)  & \text{if } \psi=\psi_1 \land \psi_2\\ 
		[\varphi] & \text{if }\psi=[\varphi] 
	\end{cases}
\]
Note that $r$ eliminates the $\ominus$ operator in distribution terms. For all $\Delta$ we have
\[
	\sem{\psi}{\Delta} = \sem{r(\psi)}{\Delta}
\]
\end{proposition}
\begin{proof}
Let $\psi=(\psi_1 \land \psi_2) \ominus p$. Then 
\[\begin{array}{rcl}
\sem{(\psi_1 \land \psi_2) \ominus p}{\Delta} & = &
\max(\sem{\psi_1 \land \psi_2}{\Delta}-p,0) \\
& = & \max(\min(\sem{\psi_1}{\Delta},\sem{\psi_2}{\Delta})-p,0) \\
& = & \min(\max(\sem{\psi_1}{\Delta}-p,0),\max(\sem{\psi_2}{\Delta}-p),0) \\
& = & \sem{(\psi_1 \ominus p) \land (\psi_2 \ominus p)}{\Delta}.
\end{array}\]

Let $\psi=[\varphi] \ominus p$. Then 
\[\begin{array}{rcl}
\sem{[\varphi] \ominus p}{\Delta} & = & \max((\sum_{s\in S}\Delta(s)\sem{\varphi}{s})-p,0)\\
& = & \max(\sum_{s\in S}\Delta(s)(\sem{\varphi}{s})-p/|S|),0) \\
& = & \sem{[\varphi \ominus p]}{\Delta}.
\end{array}\]

The other cases are trivial.
\end{proof}
}

The above metric HML induces two natural logical metrics $\bisimdlogicstrong$ and $\bisimddlogicstrong$  on states and distributions respectively, by letting
\[\begin{array}{rcl}
	\bisimdlogicstrong(s,t) & = & \sup_{\varphi\in\logicstate} | \sem{\varphi}{s} - \sem{\varphi}{t} | \\
	\bisimddlogicstrong(\Delta,\Theta) & = & \sup_{\psi\in\logicdist} | \sem{\psi}{\Delta} - \sem{\psi}{\Theta} |
\end{array}\]

\begin{example}\label{ex:1}
Consider the two probabilistic systems depicted in Figure~\ref{fig:comb.trans}.
We have the formula  $\varphi=\dia{a}\psi$ where
$\psi=[\dia{a}\top] \land [\dia{b}\top])$ and would like to know the difference $s$ and $t$ given by $\varphi$. Let 
\[
\begin{array}{rcl}
	\Delta_1 & = & 0.2\cdot\pdist{s_1} + 0.8\cdot \pdist{s_2}\\
	\Delta_2 & = & 0.8\cdot\pdist{s_5} + 0.2\cdot \pdist{s_6}\\
	\Delta_3 & = & 0.5\cdot\pdist{s_3} + 0.5\cdot \pdist{s_4}
\end{array}
\] 
Note that $\sem{\dia{a}\top}{s_{1}}=1$ and $\sem{\dia{a}\top }{s_{2}}=0$. Then $\sem{[\dia{a}\top]}{\Delta_1}=0.2\cdot\sem{\dia{a}\top}{s_1} + 0.8\cdot\sem{\dia{a}\top}{s_2}=0.2$. Similarly, $\sem{[\dia{b}\top]}{\Delta_1}=0.8$. It follows that $\sem{\psi}{\Delta_1}=\min(\sem{[\dia{a}\top]}{\Delta_1},~ \sem{[\dia{b}\top]}{\Delta_1})=0.2$. With similar arguments, we see that $\sem{\psi}{\Delta_2}=0.2$ and $\sem{\psi}{\Delta_3}=0.5$.
Therefore, we can calculate that 
\[\begin{array}{l}\sem{\varphi}{s}=\max(\sem{\psi}{\Delta_1},\sem{\psi}{\Delta_2})=0.2\\
 \sem{\varphi}{t}=\max(\sem{\psi}{\Delta_1},\sem{\psi}{\Delta_2},\sem{\psi}{\Delta_3})=0.5.
\end{array}\] 
So the difference  between $s$ and $t$ with respect to $\varphi$ is 
$|\sem{\varphi}{s} - \sem{\varphi}{t}|=0.3$. In fact we also have $\bisimdlogicstrong(s,t)=0.3$.

\begin{figure}
\begin{center}
\scalebox{0.8}{
\begin{tikzpicture}[->,>=stealth,auto,node distance=1.7cm,semithick,scale=1,every node/.style={scale=1}]
        \tikzstyle{state}=[minimum size=10pt,circle,draw,thick]
        \tikzstyle{blackdot}=[circle,fill=black, minimum
        size=3pt,inner sep=0pt]
	\node[state] (p) at (0,0) {$s$} ;
	\node[blackdot] (ppi1) at ($ (p) - (1.5,1) $){};
	\node[blackdot] (ppi2) at ($ (p) - (-1.5,1) $) {};
	\node[state] (p2) at ($ (ppi1) - (1,1) $) {$s_1$};
	\node[state] (p3) at ($ (ppi1) - (-1,1) $) {$s_2$};
	\node[state] (p4) at ($ (ppi2) - (1,1) $) {$s_5$};
	\node[state] (p5) at ($ (ppi2) - (-1,1) $) {$s_6$};
	\node[blackdot] (ppi3) at ($ (p2) - (0,1) $) {};
	\node[blackdot] (ppi4) at ($ (p3) - (0,1) $) {};
	\node[blackdot] (ppi5) at ($ (p4) - (0,1) $) {};
	\node[blackdot] (ppi6) at ($ (p5) - (0,1) $) {};

	\path[->] (p) edge node [above] {{\tiny $a$}} (ppi1);
	\path[->] (p) edge node [above] {{\tiny $a$}} (ppi2);
	\path[->] (ppi1) [dotted]  edge node [left] {{\tiny $0.2$}} (p2);	
	\path[->] (ppi1) [dotted]  edge node [right] {{\tiny $0.8$}} (p3);	
	\path[->] (ppi2) [dotted]  edge node [left] {{\tiny $0.8$}} (p4);	
	\path[->] (ppi2) [dotted]  edge node [right] {{\tiny $0.2$}} (p5);	
	\path[->] (p2) edge node [left] {{\tiny $a$}} (ppi3);
	\path[->] (p3) edge node [right] {{\tiny $b$}} (ppi4);
	\path[->] (p4) edge node [left] {{\tiny $a$}} (ppi5);
	\path[->] (p5) edge node [right] {{\tiny $b$}} (ppi6);
	\path[->] (ppi3) [bend left = 85,distance=0.5cm, dotted] edge node [left] {{\tiny $1.0$}} (p2);
	\path[->] (ppi4) [bend left = 85,distance=0.5cm, dotted] edge node [left] {{\tiny $1.0$}} (p3);
	\path[->] (ppi5) [bend right = 85,distance=0.5cm, dotted] edge node [right] {{\tiny $1.0$}} (p4);
	\path[->] (ppi6) [bend right = 85,distance=0.5cm, dotted] edge node [right] {{\tiny $1.0$}} (p5);
\end{tikzpicture}
}
\vskip 5mm

\scalebox{0.8}{
\begin{tikzpicture}[->,>=stealth,auto,node distance=1.7cm,semithick,scale=1,every node/.style={scale=1}]
        \tikzstyle{state}=[minimum size=10pt,circle,draw,thick]
        \tikzstyle{blackdot}=[circle,fill=black, minimum
        size=3pt,inner sep=0pt]
	\node[state] (p) at (0,0) {$t$} ;
	\node[blackdot] (ppi1) at ($ (p) - (3,1) $) {};
	\node[blackdot] (ppi2) at ($ (p) - (0,1) $) {};
	\node[blackdot] (ppi3) at ($ (p) - (-3,1) $) {};
	\node[state] (p2) at ($ (ppi1) - (1,1) $) {$s_1$};
	\node[state] (p3) at ($ (ppi1) - (-1,1) $) {$s_2$};
	\node[state] (p4) at ($ (ppi2) - (1,1) $) {$s_3$};
	\node[state] (p5) at ($ (ppi2) - (-1,1) $) {$s_4$};
	\node[state] (p6) at ($ (ppi3) - (1,1) $) {$s_5$};
	\node[state] (p7) at ($ (ppi3) - (-1,1) $) {$s_6$};
	\node[blackdot] (ppi4) at ($ (p2) - (0,1) $) {};
	\node[blackdot] (ppi5) at ($ (p3) - (0,1) $) {};
	\node[blackdot] (ppi6) at ($ (p4) - (0,1) $) {};
	\node[blackdot] (ppi7) at ($ (p5) - (0,1) $) {};
	\node[blackdot] (ppi8) at ($ (p6) - (0,1) $) {};
	\node[blackdot] (ppi9) at ($ (p7) - (0,1) $) {};

	\path[->] (p) edge node [above] {{\tiny $a$}} (ppi1);
	\path[->] (p) edge node [right] {{\tiny $a$}} (ppi2);
	\path[->] (p) edge node [above] {{\tiny $a$}} (ppi3);
	\path[->] (ppi1) [dotted]  edge node [left] {{\tiny $0.2$}} (p2);	
	\path[->] (ppi1) [dotted]  edge node [right] {{\tiny $0.8$}} (p3);	
	\path[->] (ppi2) [dotted]  edge node [left] {{\tiny $0.5$}} (p4);	
	\path[->] (ppi2) [dotted]  edge node [right] {{\tiny $0.5$}} (p5);	
	\path[->] (ppi3) [dotted]  edge node [left] {{\tiny $0.8$}} (p6);	
	\path[->] (ppi3) [dotted]  edge node [right] {{\tiny $0.2$}} (p7);	
	\path[->] (p2) edge node [left] {{\tiny $a$}} (ppi4);
	\path[->] (p3) edge node [right] {{\tiny $b$}} (ppi5);
	\path[->] (p4) edge node [left] {{\tiny $a$}} (ppi6);
	\path[->] (p5) edge node [right] {{\tiny $b$}} (ppi7);
	\path[->] (p6) edge node [left] {{\tiny $a$}} (ppi8);
	\path[->] (p7) edge node [right] {{\tiny $b$}} (ppi9);
	\path[->] (ppi4) [bend left = 85,distance=0.5cm, dotted] edge node [left] {{\tiny $1.0$}} (p2);
	\path[->] (ppi5) [bend left = 85,distance=0.5cm, dotted] edge node [left] {{\tiny $1.0$}} (p3);
	\path[->] (ppi6) [bend right = 85,distance=0.5cm, dotted,pos=0.65] edge node [right] {{\tiny $1.0$}} (p4);
	\path[->] (ppi7) [bend left = 85,distance=0.5cm, dotted,pos=0.35] edge node [left] {{\tiny $1.0$}} (p5);
	\path[->] (ppi8) [bend right = 85,distance=0.5cm, dotted] edge node [right] {{\tiny $1.0$}} (p6);
	\path[->] (ppi9) [bend right = 85,distance=0.5cm, dotted] edge node [right] {{\tiny $1.0$}} (p7);
\end{tikzpicture}
}
\caption{ $\bisimdlogicstrong(s,t)=0.3$ 
}\label{fig:comb.trans}
\end{center}
\end{figure}

\end{example}

\begin{example}
At first sight the following two equations seem to be sound.
\[
\sem{[\varphi]\ominus p}{\Delta} = \sem{[\varphi\ominus p]}{\Delta}
 \qquad\mbox{and}\qquad
\sem{\psi}{\sum_ip_i\Delta_i} = \sum_i p_i(\sem{\psi}{\Delta_i})
\]
However, in general they do not hold, as witnessed by the counterexamples below. Let $\varphi=\dia{b}\top$, $\psi=[\varphi]\ominus 0.5$ and the distribution $\Delta_1$ be the same as in Example~\ref{ex:1}. Then we have
\[\begin{array}{rcl}
\sem{[\varphi]\ominus 0.5}{\Delta_1} & = & \max(\sem{[\varphi]}{\Delta_1}-0.5,\; 0) \\ & = &  \max(0.2 \sem{[\dia{b}\top]}{\pdist{s_1}} + 0.8 \sem{[\dia{b}\top]}{\pdist{s_2}} -0.5,\; 0) \\
 & = & \max(0.2 \cdot 0 + 0.8 \cdot 1 -0.5,\; 0) \\
& = & 0.3 \\ \\

\sem{[\varphi\ominus 0.5]}{\Delta_1} & = & 0.2 \sem{\varphi\ominus 0.5}{s_1} + 0.8 \sem{\varphi\ominus 0.5}{s_2}\\
 & = & 0.2 \max(\sem{\varphi}{s_1}-0.5,\; 0) + 0.8 \max(\sem{\varphi}{s_2}-0.5,\; 0) \\
 & = & 0.2 \max(0-0.5,\; 0) + 0.8 \max(1-0.5,\; 0)\\
 & = & 0.4 \\ \\

0.2 \sem{\psi}{\pdist{s_1}} + 0.8 \sem{\psi}{\pdist{s_2}}
& =  & 0.2 \sem{[\varphi]\ominus 0.5}{\pdist{s_1}} + 0.8 \sem{[\varphi]\ominus 0.5}{\pdist{s_2}}\\
 & = & 0.2 \max(\sem{[\varphi]}{\pdist{s_1}}- 0.5,\; 0) + 0.8 \max(\sem{[\varphi]}{\pdist{s_2}}- 0.5,\; 0) \\
& = & 0.2 \max(0-0.5,\; 0) + 0.8 \max(1-0.5,\; 0)\\
& = & 0.4
\end{array}\]
So we see that $\sem{[\varphi]\ominus 0.5}{\Delta_1}\not=\sem{[\varphi\ominus 0.5]}{\Delta_1}$ and $\sem{\psi}{\Delta_1}\not=
0.2 \sem{\psi}{\pdist{s_1}} + 0.8 \sem{\psi}{\pdist{s_2}}$
\end{example}

In what follows we will show that the logic $\logic$ precisely captures the bisimilarity metric $\bisimdstrong$: the metric $\bisimdlogicstrong$ defined by state formulae coincides with $\bisimdstrong$ and the metric $\bisimddlogicstrong$ defined by distribution formulae coincides with $\Kantorovich(\bisimdstrong)$, the lifted form of $\bisimdstrong$. The two properties are entangled because state formulae and distribution formulae are not independent.

\begin{lemma}\label{lem:dldst}
\begin{enumerate}
\item $\bisimdlogicstrong \sqsubseteq \bisimdstrong$
\item $\bisimddlogicstrong \sqsubseteq \Kantorovich(\bisimdstrong)$
\end{enumerate}
\end{lemma}
\begin{proof}
We show the two statements simultaneously by structural induction on formulae. For any two states $s,t\in S$ and distributions $\Delta_1,\Delta_2\in \dist{S}$, we prove that
\begin{enumerate}
\item[(1)] $| \sem{\varphi}{s} - \sem{\varphi}{t} | \leq \bisimdstrong(s,t)$ for all $\varphi\in \logicstate$;
\item[(2)] $| \sem{\psi}{\Delta_1} - \sem{\psi}{\Delta_2} | \leq \Kantorovich(\bisimdstrong)(\Delta_1,\Delta_2)$ for all $\psi\in \logicdist$.
\end{enumerate}
We first analyze the structure of $\varphi$ in (1).
\begin{itemize}
\item $\varphi \equiv \top$. Then it is trivial to see that $| \sem{\varphi}{s} - \sem{\varphi}{t} | = |1-1| = 0 \leq \bisimdstrong(s,t)$.

\item $\varphi\equiv \lnot\varphi'$. Then $| \sem{\varphi}{s} - \sem{\varphi}{t} | = | \sem{\varphi'}{t} - \sem{\varphi'}{s} | \leq \bisimdstrong(s,t)$ where the inequality holds by induction.

\item $\varphi \equiv \varphi'\ominus p$. There are four subcases and we consider one of them. Suppose $\sem{\varphi'}{s}>p$ and $\sem{\varphi'}{t}\leq p$, then 
$|\sem{\varphi}{s}-\sem{\varphi}{t}| = |\sem{\varphi'}{s}-p| \leq |\sem{\varphi'}{s} - \sem{\varphi'}{t}| \leq \bisimdstrong(s,t)$ by induction.

\item $\varphi\equiv \varphi_1 \land \varphi_2$. Without loss of generality we assume that $\sem{\varphi}{s} \geq \sem{\varphi}{t}$. There are two possibilities:
\begin{itemize}
\item If $\sem{\varphi_1}{t} \leq \sem{\varphi_2}{t}$, then $\sem{\varphi}{s} - \sem{\varphi}{t} \leq \sem{\varphi_1}{s} - \sem{\varphi_1}{t} \leq \bisimdstrong(s,t)$, where the last inequality holds by induction.

\item Symmetrically, if  $\sem{\varphi_2}{t} \leq \sem{\varphi_1}{t}$, then $\sem{\varphi}{s} - \sem{\varphi}{t} \leq \sem{\varphi_2}{s} - \sem{\varphi_2}{t} \leq \bisimdstrong(s,t)$.
\end{itemize}

\item $\varphi\equiv \dia{a}\psi$. We consider the non-trivial case that both $s$ and $t$ can perform action $a$. (If either of the two states cannot perform action $a$, the expected result is straightforward.) Let $\Delta_1$ be the distribution such that $s \trans[a] \Delta_1$ and $\sem{\dia{a}\psi}{s} = \sem{\psi}{\Delta_1}$. 
Since $\bisimdstrong$ is a state-based bisimulation metric, by definition there exists some $\Delta_2$ such that $t\trans[a]\Delta_2$ and $\Kantorovich(\bisimdstrong)(\Delta_1,\Delta_2) \leq \bisimdstrong(s,t)$. Without loss of generality we assume that $\sem{\varphi}{s} \geq \sem{\varphi}{t}$. It follows that
\[\begin{array}{rl}
& \sem{\varphi}{s} - \sem{\varphi}{t} \\
= & \sem{\psi}{\Delta_1} - \max_{s \trans[a] {\Delta'}} \sem{\psi}{\Delta'} \\
\leq & \sem{\psi}{\Delta_1} - \sem{\psi}{\Delta_2} \\
\leq & \Kantorovich(\bisimdstrong)(\Delta_1,\Delta_2) \qquad\mbox{by induction on $\psi$}\\
\leq & \bisimdstrong(s,t) 
\end{array}\]
\end{itemize}

Then we analyze the structure of $\psi$ in (2).
\begin{itemize}
\item $\psi = \psi_1\land \psi_2$. Without loss of generality we assume that $\sem{\psi}{\Delta_1} \geq \sem{\psi}{\Delta_2}$. There are two possibilities:
\begin{itemize}
\item If $\sem{\psi_1}{\Delta_2} \leq \sem{\psi_2}{\Delta_2}$, then $\sem{\psi}{\Delta_1} - \sem{\psi}{\Delta_2} \leq \sem{\psi_1}{\Delta_1} - \sem{\psi_1}{\Delta_2} \leq \Kantorovich(\bisimdstrong)(\Delta_1,\Delta_2)$, where the last inequality holds by induction.

\item Symmetrically, if  $\sem{\psi_2}{\Delta_2} \leq \sem{\psi_1}{\Delta_2}$, then $\sem{\psi}{\Delta_1} - \sem{\psi}{\Delta_2} \leq \sem{\psi_2}{\Delta_1} - \sem{\psi_2}{\Delta_2} \leq \Kantorovich(\bisimdstrong)(\Delta_1,\Delta_2)$.
\end{itemize}

\item $\psi = \psi'\ominus p$ for some $p$ in $[0,1]$. There are four subcases and we consider one of them. Suppose $\sem{\psi'}{\Delta_1}>p$ and $\sem{\psi'}{\Delta_2}\leq p$, then 
$|\sem{\psi}{\Delta_1}-\sem{\psi}{\Delta_2}| = |\sem{\psi'}{\Delta_1}-p| \leq |\sem{\psi'}{\Delta_1} - \sem{\psi'}{\Delta_2}| \leq \Kantorovich(\bisimdstrong)(\Delta_1,\Delta_2)$ by induction.

\item $\psi = [\varphi]$ for some $\varphi\in \logicstate$. Without loss of generality we assume that $\sem{\varphi}{s} \geq \sem{\varphi}{t}$. We infer that
\[\begin{array}{ll}
& \sem{\psi}{\Delta_1} - \sem{\psi}{\Delta_2}\\
= & \sem{[\varphi]}{\Delta_1} - \sem{[\varphi]}{\Delta_2}\\
= & \sum_{u\in S}(\Delta_1(u)-\Delta_2(u)) \sem{\varphi}{u} \\
\leq & \max\{\sum_{u\in S}(\Delta_1(u)-\Delta_2(u)) x_u \mid x_u, x_{u'}\in [0,1] \land x_u-x_{u'}\leq \bisimdstrong(u,u')\} \\
= & \Kantorovich(\bisimdstrong)(\Delta_1,\Delta_2) 
\end{array}\]
where the last equality holds because of the Kantorovich-Rubinstein duality theorem \cite{KR58,BW01a} and the last inequality holds because for any states $u,u'\in S$ we have $\sem{\varphi}{u}, \sem{\varphi}{u'}\in [0,1]$ and $|\sem{\varphi}{u} - \sem{\varphi}{u'}| \leq \bisimdstrong(u,u')$ by induction on $\varphi$.
\end{itemize}
\end{proof}

\begin{lemma}\label{lem:dist.logic} 
$\Kantorovich(\bisimdlogicstrong) \sqsubseteq \bisimddlogicstrong$
\end{lemma}
\begin{proof}
Let  $\Delta_1,\Delta_2$ be any two distributions in $\logicdist$.
We show that 
\[\Kantorovich(\bisimdlogicstrong)(\Delta_1,\Delta_2)~\leq~\sup_{\psi\in\logicdist}|\sem{\psi}{\Delta_1}-\sem{\psi}{\Delta_2}|,\] 
using an idea inspired by \cite{DJGP02}. Let $L(\Delta_1,\Delta_2)$ be the optimal value of the following linear program
\[\begin{array}{rl}
\max \sum_{s\in S}(\Delta_1(s)-\Delta_2(s)) x_s, \text{ subject to } & 0\leq x_s\leq 1 \\ 
&\forall s,t\in S\colon\ x_s-x_{t}\leq \bisimdlogicstrong(s,t)
\end{array}\]

Suppose $\{k_s\}_{s\in S}$ be a set of real numbers that maximizes the above linear program to reach $L(\Delta_1,\Delta_2)$. Let $e=\min\{1-k_t\mid k_t < 1 \mbox{ and } t\in S\}$ and
$\epsilon > 0$ be any positive real number smaller than $e$. Hence, if $t\in S$ and $k_t< 1$ then
\begin{equation}\label{eq:ktlessone}
	k_t+\epsilon < 1.
\end{equation}
We construct some formula $\psi$ such that 
\begin{equation}\label{eq:l}
L(\Delta_1,\Delta_2)-\epsilon ~<~ \sem{\psi}{\Delta_1}-\sem{\psi}{\Delta_2}.
\end{equation} 

 For any $s,t\in S$, we distinguish two cases:
\begin{enumerate}
\item If $k_s > k_t$, then $0< k_s - k_t \leq \bisimdlogicstrong(s,t)$. It is easy to see that there exists some formula $\varphi_{st}$ such that 
\begin{equation}\label{eq:ks}
k_s - k_t < \sem{\varphi_{st}}{s}-\sem{\varphi_{st}}{t} +\epsilon.
\end{equation} 
or equivalently $\sem{\varphi_{st}}{t} - \sem{\varphi_{st}}{s} + k_s < k_t +\epsilon$.
We define a new formula
\[\varphi'_{st}\; = \; \left\{\begin{array}{ll}
\varphi_{st} \ominus (\sem{\varphi_{st}}{s} - k_s) \ \ \ & \mbox{if $\sem{\varphi_{st}}{s}> k_s$}\\
\varphi_{st} \rplus (k_s - \sem{\varphi_{st}}{s}) & \mbox{otherwise.}
\end{array}\right.
\]
Let us compare $\varphi'_{st}$ with $k_t$.
\begin{enumerate}
\item If $\sem{\varphi_{st}}{s}> k_s$, then
  \[\begin{array}{rcl}
   \sem{\varphi'_{st}}{t} & = & \max(\sem{\varphi_{st}}{t} - \sem{\varphi_{st}}{s} + k_s, 0) \\
& < & \max(k_t+\epsilon, 0) \qquad\mbox{by (\ref{eq:ks})}\\
& = & k_t+\epsilon
  \end{array}\]
\item Otherwise, 
   $\sem{\varphi'_{st}}{t}  =  \min(\sem{\varphi_{st}}{t} + k_s - \sem{\varphi_{st}}{s}, 1)$. By  (\ref{eq:ks}) we infer that $\sem{\varphi_{st}}{t} + k_s - \sem{\varphi_{st}}{s} < k_t + \epsilon$. Since $k_t < k_s \leq 1$ 
we infer by (\ref{eq:ktlessone}) that $k_t+\epsilon < 1$ and thus $\sem{\varphi'_{st}}{t} < k_t+\epsilon$. 
\end{enumerate}
In both (a) and (b) we have $\sem{\varphi'_{st}}{t} < k_t+\epsilon$, and it is also easy to see that $\sem{\varphi'_{st}}{s}=k_s$.

\item If $k_s\leq k_t$, then we simply set $\varphi'_{st}$ to be the formula $\const{k_s}$. As in the last case, we have $\sem{\varphi'_{st}}{s}=k_s$ and $\sem{\varphi'_{st}}{t}=k_s \leq k_t < k_t+\epsilon$.
\end{enumerate}

In summary, the above reasoning says that for any $s,t\in S$ we can construct a formula $\varphi'_{st}$ such that
$\sem{\varphi'_{st}}{s}=k_s$ and $\sem{\varphi'_{st}}{t} < k_t+\epsilon$. 
Now let us define $\varphi'_s=\bigwedge_{t\in S}\varphi'_{st}$. 
It is easy to see that $\sem{\varphi'_s}{s}=k_s$ and $\sem{\varphi'_s}{t}< k_t+\epsilon$ for all $t\in S$. The latter implies $\max\{\sem{\varphi'_s}{t} \mid s,t\in S\} < k_t+\epsilon$.
Then define $\varphi=\bigvee_{s\in S}\varphi'_s$. For all $t\in S$, we have 
\[k_t ~=~ \sem{\varphi'_t}{t} ~\leq~ \sem{\varphi}{t} ~=~ \max\{\sem{\varphi'_s}{t} \mid s,t\in S\} ~<~ k_t+\epsilon.\] 
Finally, we define $\psi=[\varphi]$.
It follows that 
\[\begin{array}{rcl}
\sem{\psi}{\Delta_1} -\sem{\psi}{\Delta_2} & = & \sem{[\varphi]}{\Delta_1} - \sem{[\varphi]}{\Delta_2} \\
& = & \sum_{t\in S}\Delta_1(t)\cdot \sem{\varphi}{t} -\sum_{t\in S}\Delta_2(t))\cdot\sem{\varphi}{t} \\
& \geq & \sum_{t\in S}\Delta_1(t)\cdot k_t - \sum_{t\in S}\Delta_2(t))\cdot\sem{\varphi}{t} \\
& > & \sum_{t\in S}\Delta_1(t)\cdot k_t - \sum_{t\in S}\Delta_2(t))\cdot (k_t +\epsilon)\\
& = & \sum_{t\in S}(\Delta_1(t) - \Delta_2(t)) \cdot k_t - \sum_{t\in S}\Delta_2(t)\cdot\epsilon \\
& = & L(\Delta_1,\Delta_2) - \epsilon
\end{array}\]
as required.
\end{proof}

The above property will be used to prove the following lemma.
\begin{lemma}\label{lem:stl}
$\bisimdstrong \sqsubseteq \bisimdlogicstrong$
\end{lemma}
\begin{proof}
We show that $\bisimdlogicstrong$ is a state-based bisimulation metric. Let $s,t$ be any two states in $S$ and $\epsilon$ be any real number in the interval $[0,1)$ with $\bisimdlogicstrong(s,t)\leq\epsilon$. Assume that $s\trans[a]\Delta_1$ is an arbitrarily chosen transition from $s$. Then state $t$ must be able to perform action $a$ too. Otherwise it is easy to see that $\bisimdlogicstrong(s,t)=1>\epsilon$, which contradicts our assumption above. We need to show that there exists some transition $t\trans[a]\Delta_2$ with $\Kantorovich(\bisimdlogicstrong)(\Delta_1,\Delta_2)\leq\epsilon$. 
Suppose for a contradiction that no $a$-transition from $t$ satisfies this condition.
In other words, for each $\Delta^i_2$ with $t\trans[a]\Delta^i_2$ we have $\Kantorovich(\bisimdlogicstrong)(\Delta_1,\Delta^i_2) > \epsilon$. By Lemma~\ref{lem:dist.logic}, this means $\bisimddlogicstrong(\Delta_1,\Delta^i_2) > \epsilon$. Then
there must exist some formula $\psi^i_2 \in \logicdist$ such that 
$|\sem{\psi^i_2}{\Delta_1}-\sem{\psi^i_2}{\Delta^i_2}| > \epsilon$.
Furthermore, we can strengthen this condition to the following one
\begin{equation}\label{eq:psii}
	\sem{\psi^i_2}{\Delta_1}-\sem{\psi^i_2}{\Delta^i_2} > \epsilon
\end{equation}
because we can take the formula $\lnot\psi^i_2$ in place of $\psi^i_2$ in the case that $\sem{\psi^i_2}{\Delta_1} < \sem{\psi^i_2}{\Delta_2}$. 
Let $\varphi=\dia{a}\bigwedge_i(\psi^i_2 \ominus \sem{\psi^i_2}{\Delta^i_2})$. We infer that
\[\begin{array}{rcl}
\sem{\varphi}{s} & = & \max_{s\trans[a]\Delta}\sem{\bigwedge_i\psi^i_2 \ominus \sem{\psi^i_2}{\Delta^i_2}}{\Delta}\\
& \geq & \sem{\bigwedge_i(\psi^i_2 \ominus \sem{\psi^i_2}{\Delta^i_2})}{\Delta_1} \\
& = & \min_i \sem{\psi^i_2 \ominus \sem{\psi^i_2}{\Delta^i_2}}{\Delta_1} \\
& = & \sem{\psi^k_2 \ominus \sem{\psi^k_2}{\Delta^k_2}}{\Delta_1} \qquad\text{for some $k$}\\
& = & \max(\sem{\psi^k_2}{\Delta_1} - \sem{\psi^k_2}{\Delta^k_2},~0)\\
& > & \epsilon \qquad\mbox{by (\ref{eq:psii})}
\end{array}\]
On the other hand, we have
\[\begin{array}{rcl}
\sem{\varphi}{t} & = & \max_{t\trans[a]\Delta^i_2}\sem{\bigwedge_j(\psi^j_2 \ominus \sem{\psi^j_2}{\Delta^j_2})}{\Delta^i_2}\\
& = & \max_{t\trans[a]\Delta^i_2}\min_j \sem{\psi^j_2 \ominus \sem{\psi^j_2}{\Delta^j_2}}{\Delta^i_2} \\
& = & \max_{t\trans[a]\Delta^i_2}\min_j \max((\sem{\psi^j_2}{\Delta^i_2} - \sem{\psi^j_2}{\Delta^j_2}),~0) \\
& = & 0
\end{array}\]
It follows that $\bisimdlogicstrong(s,t) \geq \sem{\varphi}{s} - \sem{\varphi}{t} > \epsilon$, which gives rise to a contradiction.
\end{proof}

By combining the above three technical lemmas we obtain the following logical characterisation of the state-based bisimilairty metric.
\begin{theorem}
$\bisimdstrong=\bisimdlogicstrong$ and $\Kantorovich(\bisimdstrong)=\bisimddlogicstrong$
\hfill\qed
\end{theorem}

\begin{remark}
In the proof of Lemma~\ref{lem:stl} we have constructed the formula 
\begin{equation}\label{eq:varphi}
\varphi=\dia{a}\bigwedge_i(\psi^i_2 \ominus \sem{\psi^i_2}{\Delta^i_2})
\end{equation}
by making use of conjunction and minus connectives for distribution formulae. This happens because in the presence of non-determinism state $t$ may perform action $a$ and then evolves into one of successor distributions $\Delta^i_2$. If we confine ourselves to deterministic pLTSs, then state $t$ will have a unique successor distribution $\Delta^i_2$ and therefore (\ref{eq:varphi}) can be simplified as 
$\varphi=\dia{a}\psi^i_2$. In this case, there is no need of  conjunction and minus connectives for distribution formulae. Furthermore, if we fold $[\varphi]$ into state formulae in Definition~\ref{def:mHML}, distribution formulae can be completely dropped. In other words, for deterministic pLTSs, the state-based bisimilarity metric can be characterised by the following metric logic
\begin{equation}\label{eq:logic}	\varphi ::=  \top \ |\ \lnot\varphi \ |\ \varphi\ominus p \ |\ \varphi_1 \land \varphi_2 \ | \ \dia{a}\varphi
\end{equation}
Therefore, for deterministic pLTSs, the two-sorted logic in Definition~\ref{def:mHML} degenerates into the logic considered in \cite{DGJP04,BW05}, as expected.

In \cite{DAMRS08,CDAMR10} a bisimulation metric for game structures is characterised by a quantitative $\mu$-calculus where formulae are valuated also on states and no distribution formula is needed. This is not surprising because games are deterministic: at any state $s$, if two players have chosen their moves, say $a_1$ and $a_2$, then there is a unique distribution $\delta(s,a_1,a_2)$ to determine the probabilities of arriving at a set of destination states.
\end{remark}


\section{Distribution-Based Bisimulation Metric}\label{sec:dist}
The bisimilarity metric given in Definition~\ref{def:sb} measures the distance between two states. Alternatively, it is possible to directly define a
 metric that measures subdistributions. In order to do so, we first define a transition
relation between subdistributions.
\begin{definition}
We write $\Delta\trans[a]\Delta'$ if 
$\Delta'=\sum_{s\in\support{\Delta}}\Delta(s)\cdot\Delta_s$, where $\Delta_s$ is
determined as follows:
\begin{itemize}
\item 
  either $s  \trans[a] \Delta_s$
\item 
  or there is no $\Theta$ with $s\trans[a]\Theta$, and in this case we set
  $\Delta_s=\eDis$.
\end{itemize}
\end{definition}
Note that if $\Delta\trans[a]\Delta'$ then some (not
necessarily all) states in the support of $\Delta$ can perform action
$a$.  For example, consider the two states $s_2$ and $s_3$ in
Figure~\ref{fig:exam1}. Since $s_2\trans[c]\pdist{s_4}$ and $s_3$
cannot perform action $c$, the distribution
$\Delta=\frac{1}{2}\pdist{s_2}+\frac{1}{2}\pdist{s_3}$ can make the
transition $\Delta\trans[c]\frac{1}{2}\pdist{s_4}$ to reach the
subdistribution $\frac{1}{2}\pdist{s_4}$. 
\begin{definition}
A $1$-bounded pseudometric $d$ on $\dist{S}$ is a \emph{distribution-based bisimulation metric} if $|\;|\Delta_1|-|\Delta_2|\;| \leq d(\Delta_1,\Delta_2)$ and 
for all $\Delta_1,\Delta_2\in \dist{S}$ and $\epsilon \in [0,1)$ with $d(\Delta_1,\Delta_2) \le \epsilon$, if $\Delta_1 \trans[a] \Delta'_1$ then there exists an $\Delta_2 \trans[a] \Delta'_2$ with $d(\Delta'_1,\Delta'_2) \le \epsilon$.
\end{definition}
The smallest (wrt. $\sqsubseteq$) distribution-based bisimulation metric, notation $\bisimddist$, is called \emph{distribution-based bisimilarity metric}. 
Distribution-based bisimilarity \cite{DFD15} is the kernel of the distribution-based bisimilarity metric.

It is not difficult to see that $\bisimdstrong$ is different from $\bisimddist$, as witnessed by the following example.
\begin{example}
Consider the states in Figure~\ref{fig:exam1}.
We first observe that $\bisimddist(\pdist{s_2},\pdist{t_3})=0$ because $s_2$ and $t_3$ can match each other's action exactly. Similarly, we have $\bisimddist(\pdist{s_3},\pdist{t_4})=0$. Then it is easy to see that $\bisimddist(\frac{1}{2}\pdist{s_2}+\frac{1}{2}\pdist{s_3}, \frac{1}{2}\pdist{t_3}+\frac{1}{2}\pdist{t_4})=0$. Since $s_1\trans[b] \frac{1}{2}\pdist{s_2}+\frac{1}{2}\pdist{s_3}$ and 
$\frac{1}{2}\pdist{t_1}+\frac{1}{2}\pdist{t_2}$, we infer that
$\bisimddist(\pdist{s_1}, \frac{1}{2}\pdist{t_1}+\frac{1}{2}\pdist{t_2})=0$. It, in turn, implies $\bisimddist(\pdist{s},\pdist{t})=0$. We have already seen in Example~\ref{ex:exam1} that $\bisimdstrong(s,t)=\frac{1}{2}$. Therefore, the two distance functions $\bisimdstrong$ and $\bisimddist$ are indeed different.
\end{example}

The rest of this section is devoted to a logical characterisation of $\bisimddist$.
Consider the  metric logic $\logic^{D*}$ whose formulae are defined below
\begin{equation}\label{eq:logic1}	\psi ::=  \top \ |\ \lnot\psi \ |\ \psi\ominus p \ |\ \psi_1 \land \psi_2 \ | \ \dia{a}\psi
\end{equation}
This logic is the same as that defined in (\ref{eq:logic}) except that now we only have distribution formulae. We will show that this logic can capture the distribution-based bisimilarity metric.

\begin{definition}
A formula $\psi \in \logic^{D*}$ evaluates in $\Delta \in \dist{S}$ as follows:
\begin{itemize}[\label={}]
	\item $\displaystyle \sem{\top}{\Delta}=|\Delta|$,\\[-0.5ex]
	\item $\displaystyle \sem{\lnot \psi}{\Delta}=1-\sem{\psi}{\psi}$,\\[-0.5ex]
	\item $\displaystyle \sem{\psi\ominus p}{\Delta}=\max(\sem{\psi}{\psi}-p,~ 0)$,\\[-0.5ex]
	\item $\displaystyle \sem{\psi_1 \land \psi_2}{\Delta}=\min(\sem{\psi_1}{\Delta},\sem{\psi_2}{\Delta})$,\\[-0.5ex]
	\item $\displaystyle \sem{\dia{a}\psi}{\Delta}=\max_{\Delta \trans[a] \Delta'} \sem{\psi}{\Delta'}$,\\[-0.5ex]
\end{itemize}
\end{definition}
This induces a natural logical metric $\bisimddistl$ over subdistributions defined by \[\bisimddistl(\Delta,\Theta) = \sup_{\psi\in \logic^{D}}|\Op{\psi}(\Delta) - \Op{\psi}(\Theta)|\]
It turns out that $\bisimddistl$ coincides with $\bisimddist$. We split the proof of this coincidence result into two parts, to show that one metric is dominated by the other and vice versa.
\begin{lemma}\label{lem:dldstP}
 $\bisimddistl \sqsubseteq \bisimddist$
\end{lemma}
\begin{proof}
Similar to the proof of Lemma~\ref{lem:dldst}.
We proceed by structural induction on formulae. For any two subdistributions $\Delta_1,\Delta_2\in \dist{S}$, we prove that
\[ | \sem{\psi}{\Delta_1} - \sem{\psi}{\Delta_2} | \leq \bisimddist(\Delta_1,\Delta_2)\] for all $\psi\in \logic^{D*}$.

We first analyze the structure of $\psi$.
\begin{itemize}
\item $\varphi \equiv \top$. Then it is trivial to see that $| \sem{\psi}{\Delta_1} - \sem{\psi}{\Delta_2} | = |\,|\Delta_1|-|\Delta_2|\,| \leq \bisimddist(\Delta_1,\Delta_2)$.

\item $\psi\equiv \lnot\psi'$. Then $| \sem{\psi}{\Delta_1} - \sem{\psi}{\Delta_2} | = | \sem{\psi'}{\Delta_2} - \sem{\psi'}{\Delta_1} | \leq \bisimddist(\Delta_1,\Delta_2)$ where the inequality holds by induction.

\item $\psi \equiv \psi'\ominus p$. There are four subcases and we consider one of them. Suppose $\sem{\psi'}{\Delta_1}>p$ and $\sem{\psi'}{\Delta_2}\leq p$, then 
$|\sem{\psi}{\Delta_1}-\sem{\psi}{\Delta_2}| = |\sem{\psi'}{\Delta_1}-p| \leq |\sem{\psi'}{\Delta_1} - \sem{\psi'}{\Delta_2}| \leq \bisimddist(\Delta_1,\Delta_2)$ by induction.

\item $\psi\equiv \psi_1 \land \psi_2$. Without loss of generality we assume that $\sem{\psi}{\Delta_1} \geq \sem{\psi}{\Delta_2}$. There are two possibilities:
\begin{itemize}
\item If $\sem{\psi_1}{\Delta_2} \leq \sem{\psi_2}{\Delta_2}$, then $\sem{\psi}{\Delta_1} - \sem{\psi}{\Delta_2} \leq \sem{\psi_1}{\Delta_1} - \sem{\psi_1}{\Delta_2} \leq \bisimddist(\Delta_1,\Delta_2)$, where the last inequality holds by induction.

\item Symmetrically, if  $\sem{\psi_2}{\Delta_2} \leq \sem{\psi_1}{\Delta_2}$, then $\sem{\psi}{\Delta_1} - \sem{\psi}{\Delta_2} \leq \sem{\psi_2}{\Delta_1} - \sem{\psi_2}{\Delta_2} \leq \bisimddist(\Delta_1,\Delta_2)$.
\end{itemize}

\item $\psi\equiv \dia{a}\psi'$.  Let $\Delta'_1$ be the distribution such that $\Delta_1 \trans[a] \Delta'_1$ and $\sem{\dia{a}\psi'}{\Delta_1} = \sem{\psi'}{\Delta'_1}$. 
Since $\bisimddist$ is a distribution-based bisimulation metric, by definition there exists some $\Delta'_2$ such that $\Delta_2\trans[a]\Delta'_2$ and $\bisimddist(\Delta'_1,\Delta'_2) \leq \bisimddist(\Delta_1,\Delta_2)$. Without loss of generality we assume that $\sem{\psi}{\Delta_1} \geq \sem{\psi}{\Delta_2}$. It follows that
\[\begin{array}{rl}
& \sem{\psi}{\Delta_1} - \sem{\psi}{\Delta_2} \\
= & \sem{\psi'}{\Delta'_1} - \max_{\Delta_2 \trans[a] {\Delta''_2}} \sem{\psi'}{\Delta''_2} \\
\leq & \sem{\psi'}{\Delta'_1} - \sem{\psi'}{\Delta'_2} \\
\leq & \bisimddist(\Delta'_1,\Delta'_2) \qquad\mbox{by induction on $\psi'$}\\
\leq & \bisimddist(\Delta_1,\Delta_2) 
\end{array}\]
\end{itemize}
\end{proof}

\begin{lemma}\label{lem:stlP}
$\bisimddist \sqsubseteq \bisimddistl$
\end{lemma}
\begin{proof}
Similar to the proof of Lemma~\ref{lem:stl}.
We show that $\bisimddistl$ is a distribution-based bisimulation metric. Let $\Delta_1,\Delta_2$ be any two subdistributions in $\dist{S}$ and $\epsilon$ be any real number in the interval $[0,1)$ with $\bisimddistl(\Delta_1,\Delta_2)\leq\epsilon$. Assume that $\Delta_1\trans[a]\Delta'_1$ is an arbitrarily chosen transition from $\Delta_1$.  We need to show that there exists some transition $\Delta_2\trans[a]\Delta'_2$ with $\bisimddistl(\Delta'_1,\Delta'_2)\leq\epsilon$. 
Suppose for a contradiction that no $a$-transition from $\Delta_2$ satisfies this condition.
In other words, for each $\Delta^i_2$ with $\Delta_2\trans[a]\Delta^i_2$ we have $\bisimddistl(\Delta'_1,\Delta^i_2) > \epsilon$.  Then
there must exist some formula $\psi^i_2 \in \logicdist$ such that 
$|\sem{\psi^i_2}{\Delta'_1}-\sem{\psi^i_2}{\Delta^i_2}| > \epsilon$.
Furthermore, we can strengthen this condition to the following one
\begin{equation}\label{eq:psii1}
	\sem{\psi^i_2}{\Delta'_1}-\sem{\psi^i_2}{\Delta^i_2} > \epsilon
\end{equation}
because we can take the formula $\lnot\psi^i_2$ in place of $\psi^i_2$ in the case that $\sem{\psi^i_2}{\Delta'_1} < \sem{\psi^i_2}{\Delta_2}$. 
Let $\varphi=\dia{a}\bigwedge_i(\psi^i_2 \ominus \sem{\psi^i_2}{\Delta^i_2})$. We infer that
\[\begin{array}{rcl}
\sem{\varphi}{\Delta_1} & = & \max_{\Delta_1\trans[a]\Delta}\sem{\bigwedge_i\psi^i_2 \ominus \sem{\psi^i_2}{\Delta^i_2}}{\Delta}\\
& \geq & \sem{\bigwedge_i(\psi^i_2 \ominus \sem{\psi^i_2}{\Delta^i_2})}{\Delta'_1} \\
& = & \min_i \sem{\psi^i_2 \ominus \sem{\psi^i_2}{\Delta^i_2}}{\Delta'_1} \\
& = & \sem{\psi^k_2 \ominus \sem{\psi^k_2}{\Delta^k_2}}{\Delta'_1} \qquad\text{for some $k$}\\
& = & \max(\sem{\psi^k_2}{\Delta'_1} - \sem{\psi^k_2}{\Delta^k_2},~0)\\
& > & \epsilon \qquad\mbox{by (\ref{eq:psii1})}
\end{array}\]
On the other hand, we have
\[\begin{array}{rcl}
\sem{\psi}{\Delta_2} & = & \max_{\Delta_2\trans[a]\Delta^i_2}\sem{\bigwedge_j(\psi^j_2 \ominus \sem{\psi^j_2}{\Delta^j_2})}{\Delta^i_2}\\
& = & \max_{\Delta_2\trans[a]\Delta^i_2}\min_j \sem{\psi^j_2 \ominus \sem{\psi^j_2}{\Delta^j_2}}{\Delta^i_2} \\
& = & \max_{\Delta_2\trans[a]\Delta^i_2}\min_j \max((\sem{\psi^j_2}{\Delta^i_2} - \sem{\psi^j_2}{\Delta^j_2}),~0) \\
& = & 0
\end{array}\]
It follows that $\bisimddistl(\Delta_1,\Delta_2) \geq \sem{\psi}{\Delta_1} - \sem{\psi}{\Delta_2} > \epsilon$, which gives rise to a contradiction.
\end{proof}

By combining the previous two lemmas, we arrive at the following logical characterisation of the distribution-based bisimulation metric.
\begin{theorem}
$\bisimddist = \bisimddistl$
\hfill\qed
\end{theorem}

\section{Related work}
\label{sec:relwork}

Metrics for probabilistic transition systems are first suggested by 
Giacalone {\em et al.}  to formalize
a notion of distance between processes.  They are used also in
\cite{KN96,Nor97} to give denotational semantics for reactive
models. De Vink and Rutten \cite{dVR99} show that discrete
probabilistic transition systems can be viewed as coalgebras.  They
consider the category of complete ultrametric spaces. Similar
ultrametric spaces are considered by den Hartog in \cite{Har02}. 

Metrics for deterministic systems are extensively studied.
Desharnais {\em et al.} \cite{DGJP04} propose a logical pseudometric
for labelled Markov chains, which is a reactive model of probabilistic
systems.  A similar pseudometric 
 is defined by van Breugel and Worrell
\cite{BW01b} via the terminal coalgebra of a functor based on a metric
on the space of Borel probability measures. The metric of \cite{DGJP04,BW05} works for
continuous probabilistic transition systems, while in this work we
concentrate on discrete systems.
Interestingly, van
Breugel and Worrell \cite{BW01a} also present a polynomial-time
algorithm to approximate their coalgebraic distances. Furthermore, van Breugel \emph{et al.} propose an algorithm to approximate a behavioural pseudometric without discount \cite{BSW08}.  In \cite{FPP11} a sampling algorithm for calculating bisimulation distances in Markov decision processes is shown to have good performance.
In \cite{DAMRS07,DAMRS08} the probabilistic bisimulation metric on game structures is characterised by a quantitative $\mu$-caluclus. Algorithms for game metrics are proposed in \cite{CDAMR10,Ram10}. 

Metrics for nondeterministic probabilistic systems are considered in
\cite{DJGP02}, where Desharnais {\em et al.} deal with labelled concurrent
Markov chains (similar to pLTSs, this model can be captured by the simple probabilistic
automata of \cite{Seg95a}).  They show that the greatest fixed point
of a monotonic function on pseudometrics corresponds to the weak
probabilistic bisimilarity of \cite{PLS00}. 

In \cite{SDC07} a notion of trace metric is proposed for pLTSs and a tool is developed to compute the trace metric. In \cite{DCPP06} a notion of bisimulation metric is proposed that extends the approach of \cite{DJGP02,DGJP04} to a more general framework called action-labelled quantitative transition systems.

In \cite{DAFS09} de Alfaro \emph{et al.} consider metric transition systems in which the propositions at each state are interpreted as elements of metric spaces. In that setting, trace equivalence and bisimulation give rise to linear and branching distances that can be characterised by quantitative versions of linear-time temporal logic \cite{MP91} and $\mu$-calculus \cite{Koz83}.

In \cite{Yin02} Ying proposes a
notion of  bisimulation index for the usual labelled transition systems,
by using ultrametrics on actions instead of using pseudometrics on
states. He applies bisimulation indexes on timed CCS and real time
ACP. But the deeper connection between \cite{Yin02} and our work
worths some further studies.

In \cite{GLT15} a notion of uniform continuity is proposed to be an appropriate property of probabilistic processes for compositional reasoning with respect to bisimulation metric semantics.


\section{Concluding remarks}\label{sec:conclu}
We have considered two behavioural pseudometrics for probabilistic labelled transition systems where nondeterminism and probabilities co-exist. They correspond to state-based and distribution-based bisimulations. Our modal characterisation of the state-based bisimulation metric is much simpler than an earlier proposal by Desharnais \emph{et al.} since we only use two non-expansive operators, negation and testing, rather than the general class of non-expansive operators. Our modal characterisaton of the distribution-based bisimulation metric is new. The characterisations are shown to be sound and complete. 

In the current work we have not distinguished internal actions from external ones. But it is not difficult to make the distinction and abstract away internal actions so as to introduce weak versions of bisimulation metrics. In a finite-state and finitely branching pLTS, the subdistributions reachable from a state by weak transitions is infinite but can be represented by the convex closure of a finite set \cite{Deng15}. This entails that the logical characterisation of weak bisimulation metrics would be similar to those presented here.

\bibliographystyle{alpha}
\bibliography{paper}

\newpage
\appendix
\section{Convex Bisimulation Metric}
For $\Pi \subseteq \dist{S}$ we denote by $\cc(\Pi)$ the convex closure of $\Pi$. If $\Delta\in\cc(\der(s,a))$ then we say $\Delta$ is a combined transition of $s$ labelled by $a$, written as $s\trans[a]_c \Delta$.
\begin{definition}[Convex bisimulation metric]
A $1$-bounded metric $d$ on $S$ is a \emph{convex bisimulation metric} if for all $s,t\in S$ and $\epsilon \in [0,1)$ with $d(s,t) \le \epsilon$, if $s \trans[a] \Delta$ then there exists a $\Delta' \in \cc(\der(t,a))$ such that $\Kantorovich(d)(\Delta,\Delta') \le \epsilon$.
\end{definition}
The smallest (wrt. $\sqsubseteq$) convex bisimulation metric, denoted by $\bisimdconvex$, is called \emph{convex bisimilarity metric}. Convex bisimilarity equivalence (also called probablistic bisimilarity) \cite{Seg95a} is the kernel of the convex bisimilarity metric.

Let us define the functor $\functorconvex\colon [0,1]^{S \times S} \to [0,1]^{S \times S}$ for $d\colon S \times S \to [0,1]$ and $s,t\in S$ by
\[
	\functorconvex(d)(s,t) = \sup_{a\in \Act} \{ \Hausdorff(\Kantorovich(d))(\cc(\der(s,a)), \cc(\der(t,a))) \}\, .
\]
It can be shown that $\functorconvex$ is monotone and its least fixed point is the convex bisimilarity metric.

Given a pLTS, we can saturate it with all possible combined transitions. In the saturated pLTS, convex bisimulation metric coincides with state-based bisimulation metric because of the following lemma.

\begin{lemma}
 $d$ is a convex bisimulation metric if and only if for all $s,t\in S$ and $\epsilon \in [0,1)$ with $d(s,t) \le \epsilon$, if $s\trans[a]_c \Delta$
 then there exists some $t\trans[a]_c \Delta'$ such that $\Kantorovich(d)(\Delta,\Delta') \le \epsilon$.
\end{lemma}
\begin{proof}
The ``if'' direction is trivial. The ``only if'' direction can be shown by making use of the following inequality:
\[\Kantorovich(d)(\sum_{i\in I} p_i\cdot\Delta_i, \sum_{i\in I} p_i\cdot\Delta'_i) 
~\leq \sum_{i\in I}p_i\cdot\Kantorovich(d)(\Delta_i,\Delta'_i)\]
for any pseudometric $d$. This holds because
\[\begin{array}{ll}
& \Kantorovich(d)(\sum_{i\in I} p_i\cdot\Delta_i, \sum_{i\in I} p_i\cdot\Delta'_i) \\
= & \max\{
\sum_{u\in S}((\sum_{i\in I}p_i\cdot\Delta_i)(u)-(\sum_{i\in I}p_i\cdot\Delta'_i)(u))x_u \mid  \\
& \hspace{6cm} x_u,\; x_{u'}\in [0,1] 
\wedge x_u-x_{u'}\leq d(u,u')\}\\
= & \max\{
\sum_{u\in S}\sum_{i\in I}p_i\cdot(\Delta_i(u)-\Delta'_i(u))x_u \mid x_u,\; x_{u'}\in [0,1] \wedge x_u-x_{u'}\leq d(u,u')\}\\
= & \max\{
\sum_{i\in I}p_i\sum_{u\in S}(\Delta_i(u)-\Delta'_i(u))x_u \mid x_u,\; x_{u'}\in [0,1] \wedge x_u-x_{u'}\leq d(u,u')\}\\
\leq & \sum_{i\in I}p_i\cdot \max\{\sum_{u\in S}(\Delta_i(u)-\Delta'_i(u))x_u \mid x_u,\; x_{u'}\in [0,1] \wedge x_u-x_{u'}\leq d(u,u')\}\\
= & \sum_{i\in I}p_i\cdot \Kantorovich(d)(\Delta_i,\Delta'_i)
\end{array}\]
\end{proof}

Consequently, a simple logical characterisation of convex bisimulation metric can be obtained by extending the logic in Section~\ref{sec:logic_characterizations} with infinitary conjunctions and by interpreting the diamand modality with combined transitions, i.e. $\sem{\dia{a}\psi}{s}=\max_{s \trans[a]_c \Delta} \sem{\psi}{\Delta}$. 
Write $\bisimdlogicconvex$ for the metric induced by this extended logic.

\begin{theorem}
$\bisimdconvex=\bisimdlogicconvex$
\hfill\qed
\end{theorem}
\section{Trace Metric}

In this section we present a notion of trace metric that enjoys a straightforward logical characterisation.
\subsection{Trace Metric}
A trace $\tr$ is a string in the set $A^*$. We write $\epsilon$ for the empty trace, and $\tr_1\cdot\tr_2$ for the concatenation of two traces $\tr_1$ and $\tr_2$. Given a finitely branching pLTS, the maximum probability that state $s$ can perform trace $\tr$ is defined as follows.
\[
\Prob(s,\tr) ~=~ \left\{\begin{array}{ll}
1 & \mbox{if $\tr=\epsilon$}\\
\max_{s\trans[a]\Delta}\sum_{t\in S}\Delta(t)\cdot \Prob(t,\tr') \qquad\qquad & \mbox{if $\tr=a\cdot\tr'$}
\end{array}\right.
\]
\begin{definition}
For any two states $s,t\in S$, the \emph{trace distance} between them is given as follows:
\[
\tracedist(s,t)~=~ \sup_{\tr\in A^*}|\Prob(s,\tr)-\Prob(t,\tr)|
\]
\end{definition}
Intuitively, the trace distance between $s$ and $t$ is the maximal difference between the maximum probabilities given by $s$ and $t$ to a same trace.

\leaveout{
\begin{definition}[\protect{\cite[Def.5]{SDC07}}]
Let $u,u'$ be traces and c be a real number lying between 0 and 1.
\[
m_c(u,u')=\left\{
\begin{array}{ll}
c^{k-1} &{\it if}~~u[k]\neq u'[k] \ and \ u[i]=u '[i] \ {\it for}~(i<k)\\
0 &otherwise\\
\end{array}
\right.
\]
where $u[i]$ means the i-th action of the trace~$u$.
\end{definition}

Let $s$ be a state in a pLTS. We let $\tdists(s)$ be the set of trace distributions generated from $s$ as defined in \cite{LSV07}.
\begin{definition}[Trace Metric]
Let $c$ be a discount factor between $0$ and $1$. A $c$-discounted pseudometric is a \emph{trace metric} w.r.t. discount factor $c$ if for all $s,t\in S$, 
$$d_c(s,t) =\Hausdorff(\Kantorovich(m_c))(\tdists(s),\tdists(t)).$$
\end{definition}
} 

\subsection{Logical Characterizations}
Let $\logic^{T}$ be the set of formulae produced by the following grammar.
\[\varphi ::= \top \mid \dia{a}{\varphi} \]
A formula evaluates in a state $s$ as follows:
\[\begin{array}{l}
\sem{\top}{s}=1\\
\sem{\dia{a}{\varphi}}{s} = \max_{s\trans[a]\Delta}\sum_{t\in S}\Delta(t)\cdot\sem{\varphi}{t}
\end{array}\]
This logic induces a natural metric $\tracedistl$.
\[
\tracedistl(s,t)~=~\sup_{\varphi\in\mathbb{O}_{tr}} | \llbracket \varphi \rrbracket (s) - \llbracket \varphi \rrbracket (t) |
\]

\begin{theorem}
$\tracedist = \tracedistl$
\end{theorem}
\begin{proof}
It is easy to see that there is a bijection $f$ between $A^*$ and $\logic^{T}$.
\[\begin{array}{rcl}
f(\epsilon) & = & \top \\
f(a\cdot\tr) & = & \dia{a}{f(\tr)}
\end{array}\]
For any state $s$ and trace $\tr$,  we can show that
$\Prob(s,\tr) = \sem{f(\tr)}{s}$
by induction on the length of $\tr$. The required result now easily follows.
\end{proof}

\leaveout{
Following the idea of \cite[Sec.~6.2]{BB08} we define $\mathbb{O}_{tr}$ as the negation-free $\mathbb{O}$ formulae. We get:

\begin{conjecture}
The trace metric is characterized by
\[
	\sup_{\varphi\in\mathbb{O}_{tr}} | \llbracket \varphi \rrbracket (s) - \llbracket \varphi \rrbracket (t) |
\]
for all $s,t\in S$.
\end{conjecture}
} 

\end{document}